\documentclass[10pt,a4paper,twocolumn]{article}
\usepackage[footnotesize]{caption}
\usepackage[english]{babel}
\usepackage{times,amsmath,amsthm,epsfig,xypic}
\usepackage[margin={0.575in,0.8in}]{geometry}
\newtheorem{thm}{Theorem}
\newtheorem{defn}{Definition}
\newtheorem{lem}{Lemma}
\newtheorem{cor}{Corollary}
\newtheorem{ex}{Example}
\newenvironment{theorem}{\begin{thm}}{\end{thm}}

\newenvironment{corollary}{\begin{cor}}{\end{cor}}
\newenvironment{example}{\begin{ex}}{\end{ex}\medskip}
\newcommand{\eps}{\varepsilon}
\newcommand{\minus}{\scalebox{0.75}[1.0]{\( - \)}}
\title{Dynamic Traitor Tracing for Arbitrary Alphabets: Divide and Conquer}
\author{Thijs Laarhoven\footnote{T. Laarhoven and J. Oosterwijk are with the Department of Mathematics and Computer Science, Eindhoven University of Technology, 5612 AZ Eindhoven, The Netherlands. \protect\\
E-mail: \{t.m.m.laarhoven,j.oosterwijk\}@tue.nl.} \and Jan-Jaap Oosterwijk\footnotemark[1] \and Jeroen Doumen\footnote{J. Doumen is with Irdeto BV, 2132 LS Hoofddorp, The Netherlands.\protect\\
E-mail: jdoumen@irdeto.com.}}
\date{\today}

\begin{document}
\maketitle

\begin{abstract}
We give a generic divide-and-conquer approach for constructing collusion-resistant probabilistic dynamic traitor tracing schemes with larger alphabets from schemes with smaller alphabets. This construction offers a linear tradeoff between the alphabet size and the codelength. In particular, we show that applying our results to the binary dynamic Tardos scheme of Laarhoven et al.\ leads to schemes that are shorter by a factor equal to half the alphabet size. Asymptotically, these codelengths correspond, up to a constant factor, to the fingerprinting capacity for static probabilistic schemes. This gives a hierarchy of probabilistic dynamic traitor tracing schemes, and bridges the gap between the low bandwidth, high codelength scheme of Laarhoven et al.\ and the high bandwidth, low codelength scheme of Fiat and Tassa.
\end{abstract}

% NOTE keywords are not used for conference papers so do not populate them
% \begin{keywords}
% keyword-1, keyword-2, keyword-3
% \end{keywords}
%

% Copyright Section: choose the appropriate copyright notice for MMSP'09, acording to the authors' affilination:

%%%%%%%%%%%%%%%%%%%%%%%%%%%%%%%%%%%%%%%%%%%%%%%%%%%%%%%%%%
%%%%%%%%%%%%%%%%%%%%%%%%%%%%%%%%%%%%%%%%%%%%%%%%%%%%%%%%%%

\section{Introduction}
\label{sec:introduction}

In this day and age of digital technology, protecting digital data from unauthorized copying and redistribution is an increasingly relevant problem. By embedding unique and imperceptible fingerprints in each copy of the content, distributors of digital content can trace pirated copies to the pirate. However, a more difficult scenario arises when several users who have purchased a copy \textit{collude} to form a coalition. When receiving their fingerprinted content, colluders can compare their copies to detect parts of the fingerprints: with the content being the same for all colluders, the differences they detect must be part of the fingerprints. Then, just assigning unique fingerprints to each user is not sufficient anymore, as the colluders may output a forgery that does not match any of their copies exactly. For this, we need \textit{collusion-resistant traitor tracing schemes}, consisting of a way to assign fingerprints to users, and an algorithm to trace a forged copy to the colluders. 

%%%%%%%%%%%%%%%%%%%%%%%%%%%%%%%%%%%%%%%%%%%%%%%%%%%%%%%%%%

\subsection{Model}

Several models have been considered for this fingerprinting game. We will focus on the \textit{restricted digit model}, where for each segment of the content, colluders always output one of their fingerprinted segments. This means that if, in some segment, all colluders receive the same fingerprint, they are forced to output this version of the content. In the literature, this is usually called the \textit{marking condition} or \textit{marking assumption}. Depending on the application, we also consider two different types of schemes. In \textit{static schemes}, for each user, the distributor generates all fingerprinted segments at once. After the pirates generate a forged copy of the whole content, the accusation algorithm has to trace this single forged copy to the colluders. In this scenario, it is impossible to guarantee that all colluders are caught, so we only require that at least one colluder is caught. In \textit{dynamic schemes} however, the content owner is more powerful, as after every single content segment he can try to catch and disconnect pirates, and adjust the fingerprints for the next segments based on the previous pirate output. When a colluder is disconnected, he no longer receives content, and we assume the other colluders continue outputting watermarked content. With dynamic schemes we are therefore able to catch all colluders, so we only say a dynamic scheme is successful if all colluders are traced. For static schemes, think of DVDs and CDs, while the dynamic setting applies to pay-tv and other live broadcasts.

A further classification of the schemes depends on the notion of security we want the scheme to achieve. For \textit{probabilistic schemes} we demand that (i) with probability at most $\eps_1$ one or more of the innocent users are caught, and (ii) with probability at most $\eps_2$ we do not catch any colluder (static schemes) or all colluders (dynamic schemes). For \textit{deterministic schemes}, we demand that $\eps_1 = \eps_2 = 0$. 

%%%%%%%%%%%%%%%%%%%%%%%%%%%%%%%%%%%%%%%%%%%%%%%%%%%%%%%%%%

\subsection{Notation}

For convenience, we introduce some more notation. We write $C$ for the set of colluders, and denote the number of colluders by $c = |C|$. We write $U$ for the set of all users, and we denote its size by $n = |U|$. For each segment $i$, at most $q$ different robust versions can be generated. We denote these by the alphabet $Q = \{0, 1, \ldots, q - 1\}$. We denote the number of successive segments that the scheme needs by the codelength $\ell$. We put the fingerprints in a matrix $X$, where each row corresponds to a user $j$ and each column to a segment $i$. To avoid confusion, throughout the paper we will consequently reserve $j$ for indexing users and $i$ for indexing positions or segments. After the code matrix $X$ is generated, the colluders get together to form a pirated copy $\vec{y}$, which due to the marking condition satisfies $y_i \in \{X_{j,i}: j \in C\}$. Then, the distributor detects this pirate output, and uses some tracing algorithm $\sigma$ on the pirate output $\vec{y}$ and code matrix $X$ to accuse a set of users $C' = \sigma(\vec{y}) \subseteq U$. A probabilistic scheme is then successful if $P(C' \not\subseteq C) \leq \eps_1$, and $P(C' \cap C = \emptyset) \leq \eps_2$ (static schemes) or $P(C \not\subseteq C') \leq \eps_2$.
\newpage

%%%%%%%%%%%%%%%%%%%%%%%%%%%%%%%%%%%%%%%%%%%%%%%%%%%%%%%%%%

\subsection{Related work} 

The results in this work are related to probabilistic dynamic traitor tracing schemes, but we will also compare our results with other dynamic or probabilistic schemes. Fiat and Tassa \cite{fiat01} describe a deterministic dynamic scheme, using an alphabet of size $q = 2c + 1$ and achieving a codelength of $\ell = c \log_2 n + c$. Since any deterministic (dynamic) scheme requires the use of an alphabet of size $q \geq c + 1$, Berkman et al.\ \cite{berkman01} then investigated whether with $q = c + 1$ one could also efficiently catch all colluders. They showed that this can be done with a codelength of $\ell = O(c^2 + c \log_2(n))$. In the area of probabilistic static schemes, the scheme of Boneh and Shaw \cite{boneh98} was the first breakthrough, achieving a codelength polynomial in the number of colluders, with an alphabet size of $q = 2$. A further improvement was given by Tardos \cite{tardos03}, who constructed a binary ($q = 2$) scheme achieving codelengths $\ell = 100 c^2 \lceil\ln(n/\eps_1)\rceil$. This scheme is widely known as the Tardos scheme. Several papers \cite{blayer08,skoric08,skoric06} then showed how the constant $100$ can be further reduced, and Laarhoven and De Weger \cite{laarhoven11} finally showed how to achieve the optimal codelength of the binary symmetric Tardos scheme, given by $\ell = (\frac{\pi^2}{2} + O(c^{-1/3})) c^2 \ln(n/\eps_1)$. Building upon this optimal static Tardos scheme, Laarhoven et al.\ \cite{laarhoven11b} showed how to construct an efficient binary dynamic Tardos scheme, which has the same asymptotic codelength (for $c \to \infty$) as the optimal static Tardos scheme, but is able to catch \textit{all} colluders with high probability. This scheme improved upon the earlier scheme of Tassa \cite{tassa05}, which uses codelengths quartic in $c$.

Besides constructions of traitor tracing schemes, several papers have also investigated theoretical bounds on the codelength needed to catch a certain number of colluders. So far, these have all focused on probabilistic static schemes. Tardos \cite{tardos03} showed that his codelength is optimal up to a constant factor. Huang and Moulin \cite{huang12b} gave the exact capacity of the binary fingerprinting game, by showing that for large $c$, a codelength of $\ell = 2 \ln(2) c^2 \ln(n/\eps_1)$ is both necessary and sufficient. This was then extended to the $q$-ary setting independently by Boesten and \v{S}kori\'{c} \cite{boesten11} and Huang and Moulin \cite{huang12}, showing that the $q$-ary capacity corresponds to $2 \ln(q) \frac{c^2}{q - 1} \ln(n/\eps_1)$ bits of information, or a codelength of $\ell = 2 \ln(2) \frac{c^2}{q - 1} \ln(n/\eps_1)$ symbols from a $q$-ary alphabet. 

%%%%%%%%%%%%%%%%%%%%%%%%%%%%%%%%%%%%%%%%%%%%%%%%%%%%%%%%%%

\subsection{Contributions and outline}

In this paper, we give a generic divide-and-conquer approach for constructing probabilistic dynamic traitor tracing schemes with large alphabets from schemes with small alphabets. This construction provides a linear tradeoff between the alphabet size $q$ and the codelength $\ell$; increasing the alphabet size by a factor $k$ leads to codes that are a factor $k$ shorter. This construction can be applied to any low-bandwidth probabilistic dynamic traitor tracing scheme, and in particular to the (binary) dynamic Tardos scheme of Laarhoven et al.\ \cite{laarhoven11b}. We show that for arbitrary alphabet sizes $q$, we obtain schemes with codelengths $\ell = (\pi^2 + O((c/q)^{-1/3} + (c \ln(\frac{q}{\eps_2})/q)^{-1/2})) \frac{c^2}{q} \ln(n/\eps_1)$, matching the fingerprinting capacity for \textit{static} $q$-ary traitor tracing schemes up to constant factors. Letting $q = O(c^{1 - \gamma})$ for some $\gamma > 0$, we get asymptotic codelengths of $\ell = (\pi^2 + O(c^{-\gamma/3})) c^{1+\gamma} \ln(n/\eps_1)$, improving upon the codelengths (and alphabet size) of Berkman et al.\ \cite{berkman01} for large $c$. As $\gamma \to 0$, these codelengths also approach the asymptotic codelengths of Fiat and Tassa \cite{fiat01}.

The outline of the paper is as follows. In Section~\ref{sec:construction}, we describe the divide-and-conquer technique to build schemes with larger alphabets from schemes with smaller alphabet sizes. Then, in Section~\ref{sec:tardos}, we apply the results to the binary dynamic Tardos scheme to obtain an efficient $q$-ary dynamic Tardos scheme, and we compare our results with previous results from the literature. Finally, in Section~\ref{sec:summary}, we give a brief summary and discussion of the results, and we mention some directions for future research.

%%%%%%%%%%%%%%%%%%%%%%%%%%%%%%%%%%%%%%%%%%%%%%%%%%%%%%%%%%
%%%%%%%%%%%%%%%%%%%%%%%%%%%%%%%%%%%%%%%%%%%%%%%%%%%%%%%%%%

\section{Construction}
\label{sec:construction}

First, let us assume that for a given alphabet size $q_0$, we have some construction mechanism $\mathcal{S}_{q_0}$ for generating $q_0$-ary dynamic traitor tracing schemes (consisting of a code $X$ and a tracing algorithm $\sigma$) for any given maximum number of colluders $c$, total number of users $n$, and for given upper bounds $\eps_1$ and $\eps_2$ on the false positive and false negative error probabilities respectively. Now, to efficiently combat collusion attacks with an alphabet of size $q = 2q_0$, we follow a two-stage process. First, we \textit{divide} (see Subsection~\ref{sub:divide}) the colluders in two groups of roughly equal size, and generate $q_0$-ary traitor tracing schemes for each group separately. Then we show how to combine these codes, such that we can \textit{conquer} (see Subsection~\ref{sub:conquer}) the whole coalition using short $q$-ary codes. Finally, in Subsection~\ref{sub:arbitrary} we show how to generalize this approach to arbitrary divisions, where $q = kq_0$ for some $k \geq 2$.

%%%%%%%%%%%%%%%%%%%%%%%%%%%%%%%%%%%%%%%%%%%%%%%%%%%%%%%%%%

\subsection{Divide}
\label{sub:divide}

Before we even start thinking about traitor tracing schemes, we consider the following problem: How can we divide the set of users $U$ in two groups $U^{(1)}, U^{(2)}$, such that each group contains the same number of colluders? Since we have no idea which of the users are the colluders, it is impossible to always do this correctly. However, if we allow some room for error, this problem can be solved quite easily. Assuming $n$ is even, we first randomly divide the set of users $U$ in two groups $U^{(1)}$ and $U^{(2)}$ of size $n/2$. Let the number of colluders in each group be denoted by $C^{(t)}$, for $t = 1,2$. Then, the number of colluders $C^{(1)}$ in $U^{(1)}$ follows a hypergeometric distribution, i.e., we are taking $n/2$ samples from a population of size $n$ with $c$ successes \textit{without replacement}. To prove that both groups contain roughly the same number of colluders, note that $\max_{t = 1,2} C^{(t)} > c/2 + a$ if and only if $|C^{(1)} - c/2| > a$. To bound the probability of the latter event, we apply a result of Chv\'{a}tal \cite{chvatal79}, which is very similar to Chernoff's bound \cite{chernoff52} for estimating tail probabilities of binomial distributions. For arbitrary values of $a > 0$, we get
\begin{align}
P\left(C^{(1)} > \frac{c}{2} + a\right) \leq e^{-2a^2/c}.
\end{align}
Furthermore, by symmetry we have $P(|C^{(1)} - \frac{c}{2}| > a) = 2P(C^{(1)} > \frac{c}{2} + a)$. So for any $\eps_2 > 0$ and $\alpha_2 = \sqrt{\ln \frac{4}{\eps_2}}$, we can take $a = \alpha_2 \sqrt{\frac{c}{2}}$ to get
\begin{align*}
P\left(\max_{t = 1,2} C^{(t)} > \frac{c}{2} + \alpha_2 \sqrt{\frac{c}{2}}\right) \leq \frac{\eps_2}{2}.
\end{align*}
Hence, each group contains $n/2$ users in total, and with probability at least $1 - \frac{\eps_2}{2}$ each group contains at most $\frac{c}{2} + \alpha_2 \sqrt{\frac{c}{2}}$ colluders.

After splitting the users in groups, for each group $t = 1, 2$, we independently generate a $q_0$-ary dynamic traitor tracing scheme $(X^{(t)}, \sigma^{(t)})$, using $\mathcal{S}_{q_0}$. For each scheme, we use a different set of $q_0$ symbols, e.g., the symbols $Q^{(1)} = \{0, \ldots, q_0-1\}$ for $U^{(1)}$, and $Q^{(2)} = \{q_0, \ldots, q-1\}$ for $U^{(2)}$. The parameters to use for generating these schemes are given below:
\begin{align*}
\left\{\begin{array}{ll}
c^{(t)} = \dfrac{c}{2} + \alpha_2 \sqrt{\dfrac{c}{2}}, \quad & \eps_1^{(t)} = \dfrac{\eps_1}{2}, \\
n^{(t)} = \dfrac{n}{2}, & \eps_2^{(t)} = \dfrac{\eps_2}{4}.
\end{array}\right\} 
\quad (t = 1,2)
\end{align*}
Here, $c^{(1)}$ is the number of colluders the scheme for $U^{(1)}$ should be resistant against, etc. This leads to two code matrices $X^{(1,2)}$ with respective codelengths $\ell^{(1,2)}$, and two tracing algorithms $\sigma^{(1,2)}$.

%%%%%%%%%%%%%%%%%%%%%%%%%%%%%%%%%%%%%%%%%%%%%%%%%%%%%%%%%%

\subsection{Conquer}
\label{sub:conquer}

Having finished the preprocessing, we now show how to weave the two $q_0$-ary codes $X^{(1,2)}$ and tracing algorithms $\sigma^{(1,2)}$ into a single $q$-ary code $X$ and tracing algorithm $\sigma$. We start by setting $i = i^{(1)} = i^{(2)} = 1$, where $i$ denotes the current position in the code $X$, and $i^{(1,2)}$ denote the current positions in $X^{(1,2)}$. 

Now, at each position $i$ and for both $t = 1,2$, we send to each user $j \in U^{(t)}$ his $i$th symbol $X_{j,i} = X^{(t)}_{j,i^{(t)}} \in Q^{(t)}$. If $i^{(t)} > \ell^{(t)}$, we expect to have caught all colluders in group $U^{(t)}$ already, so we then assign all active users in $U^{(t)}$ the empty fingerprint, denoted by $\lambda$. 

Then, after sending the $i$th symbols to all users, the coalition chooses an output symbol $y_i$. The distributor then detects this forgery, and does the following.
\begin{itemize}
  \item If $y_i \in Q^{(1)}$, we apply $\sigma^{(1)}$ to $y_i$ and the users in $U^{(1)}$. This may involve calculating accusation scores, disconnecting users etc. For users in $U^{(2)}$, nothing happens. When this is done, we increase $i^{(1)}$ by $1$.
	\item If $y_i \in Q^{(2)}$, we apply $\sigma^{(2)}$ to $y_i$ and the users in $U^{(2)}$. For users in $U^{(1)}$, we do not do anything. Afterwards, we increase $i^{(2)}$ by $1$.
	\item If $y_i = \lambda$, we terminate, and we say the scheme has failed.
\end{itemize}
Finally, we increase $i$ by $1$ and we start with sending the new round of symbols to the users. This continues until either $i > \ell^{(1)} + \ell^{(2)}$, or no pirate output is detected anymore and all colluders are caught. Theorem~\ref{thm:thm1} tells us that when the scheme terminates, with high probability we will be in the latter scenario. Before we state the theorem, we will illustrate the construction with an example.

\begin{example}
Let $U = \{1,\ldots,8\}$ and $Q = \{0,\ldots,3\}$, and suppose we want to find the (hidden) coalition $C = \{1,3,7,8\}$ of size $c = 4$. First, we divide the group of users into two groups $U^{(1)} = \{1,\ldots,4\}$ and $U^{(2)} = \{5,\ldots,8\}$, and we hope each group now contains $2$ colluders. Next, we use a construction mechanism $\mathcal{S}_2$ which allows us to generate binary dynamic traitor tracing schemes for $c = 2$ for each group, resulting in the following codes $X^{(1,2)}$ of length $\ell^{(1,2)} = 5$:
\begin{align*}
X^{(1)} = \begin{pmatrix} 0 & 1 & 0 & 1 & 1 \\ 0 & 0 & 1 & 1 & 0 \\ 1 & 0 & 1 & 0 & 1 \\ 1 & 1 & 0 & 0 & 0 \end{pmatrix}, \quad 
X^{(2)} = \begin{pmatrix} 2 & 2 & 3 & 2 & 2 \\ 3 & 2 & 2 & 3 & 3 \\ 3 & 3 & 3 & 3 & 2 \\ 2 & 3 & 2 & 2 & 3 \end{pmatrix}. 
\end{align*}
We are now ready to conquer the coalition. One by one we send the symbols, and respond to the coalition as described in Section~\ref{sub:conquer}. This leads to the following code matrix $X$ and pirate output $y$.
\setcounter{MaxMatrixCols}{20}
\begin{align*}
X &= \begin{pmatrix}
0 & \mathbf 0 & 1 & 1 & \mathbf 1 & \mathbf 0 & 1 & 1 & \mathbf 1 & \mathbf 1 & \minus \\
0 & \mathbf 0 & 0 & 0 & \mathbf 0 & \mathbf 1 & 1 & 1 & \mathbf 1 & \mathbf 0 & \lambda \\
1 & \mathbf 1 & 0 & 0 & \mathbf 0 & \mathbf 1 & 0 & 0 & \mathbf 0 & \mathbf 1 & \minus \\
1 & \mathbf 1 & 1 & 1 & \mathbf 1 & \mathbf 0 & 0 & 0 & \mathbf 0 & \mathbf 0 & \lambda \\
\mathbf 2 & 2 & \mathbf 2 & \mathbf 3 & 2 & 2 & \mathbf 2 & \mathbf 2 & \lambda & \lambda & \lambda \\
\mathbf 3 & 2 & \mathbf 2 & \mathbf 2 & 3 & 3 & \mathbf 3 & \mathbf 3 & \lambda & \lambda & \lambda \\
\mathbf 3 & 3 & \mathbf 3 & \mathbf 3 & \minus & \minus & \minus & \minus & \minus & \minus & \minus \\
\mathbf 2 & 3 & \mathbf 3 & \mathbf 2 & 2 & 2 & \mathbf 2 & \mathbf 3 & \minus & \minus & \minus 
\end{pmatrix} \\ 
y &= \; \, \begin{pmatrix}
\mathbf 3 & \mathbf 0 & \mathbf 3 & \mathbf 3 & \mathbf 1 & \mathbf 1 & \mathbf 2 & \mathbf 3 & \mathbf 0 & \mathbf 1 & \minus 
\end{pmatrix}
\end{align*}
The dashes represent disconnected users. In this case, at the end all four colluders have been caught and no innocent users were harmed in the process. Note that the bold half-columns, corresponding to segments $i$ where the pirate output $y_i$ is a symbol from that half of the alphabet, together form the codes $X^{(1)}$ and $X^{(2)}$. 
\end{example} \vspace{-0.3cm}

\begin{theorem} \label{thm:thm1}
Let the $q$-ary traitor tracing scheme be constructed as described earlier. Then, with probability at most $\eps_1$ at least one innocent user is caught, and with probability at most $\eps_2$ not all pirates are disconnected after at most $\ell = \ell^{(1)} + \ell^{(2)}$ segments.
\end{theorem}

\begin{proof}
First, note that for innocent users, nothing really changes compared to the original $q$-ary scheme. For innocent users $j \in U^{(t)}$ (for some $t = 1,2$) and positions $i$ where $y_i \notin Q^{(t)}$, the accusation algorithm does not do anything, so we only have to consider the positions $i$ where $y_i \in Q^{(t)}$. On these positions, we use the algorithm $\sigma^{(t)}$ as in the original $q_0$-ary scheme. But for the original scheme we know that if we use at most $\ell^{(t)}$ symbols, the probability that no innocent users in group $t$ are accused is at least $1 - \eps_1^{(t)}$. So the probability that none of the innocent users in any group is disconnected is at least $(1 - \eps_1^{(1)})(1 - \eps_1^{(2)}) \geq 1 - \eps_1$, as was to be shown.

For guilty users, we also use a reduction-argument to prove that with high probability, all colluders are caught. First, with probability at least $1 - \frac{\eps_2}{2}$ the number of colluders in each group is bounded from above by $\frac{c}{2} + \alpha_2 \sqrt{\frac{c}{2}}$. If this is indeed the case, then the analysis of the original schemes tells us that after at most $\ell^{(t)}$ positions, with probability at least $1 - \frac{\eps_2}{4}$ all colluders in any one of these groups is caught. Since at each segment, either $i^{(1)}$ or $i^{(2)}$ increases, at some point one of them, say $i^{(t)}$, will exceed $\ell^{(t)}$. Then we know that we will have caught all colluders with probability at least $1 - \frac{\eps_2}{4}$. So the only remaining active colluders are in the other group $U^{(t')}$, for which we also know that with probability at least $1 - \frac{\eps_2}{4}$ we will catch all colluders before $i^{(t')}$ exceeds $\ell^{(t')}$. So with probability at least $(1 - \frac{\eps_2}{2}) (1 - \frac{\eps_2}{4})^2 \geq 1 - \eps_2$, the division and both schemes are successful, and we will catch all pirates after at most $\ell^{(1)} + \ell^{(2)}$ symbols.
\end{proof}

It follows that if we have a construction mechanism $\mathcal{S}_{q_0}$ that produces schemes with codelengths $\ell_{q_0}(c,n,\eps_1,\eps_2)$ quadratic in $c$ and logarithmic in $n, \eps_1^{-1}, \eps_2^{-1}$, then the divide-and-conquer technique provides us with $q$-ary schemes (with $q = 2q_0$) achieving codelengths of
\begin{align*}
\ell_q(c,n,\eps_1,\eps_2) = 2\ell_{q_0}\left(\frac{c}{2} + \alpha_2 \sqrt{\frac{c}{2}}, \frac{n}{2}, \frac{\eps_1}{2}, \frac{\eps_2}{4}\right) \\ \approx 2\ell_{q_0}\left(\frac{c}{2}, n, \eps_1, \eps_2\right) \approx \frac{1}{2} \ell_{q_0}(c, n, \eps_1, \eps_2).
\end{align*}
The first approximation follows from $\frac{c}{2} + O(\sqrt{\frac{c}{2}}) \approx \frac{c}{2}$. So the codelength decreases by a factor of approximately $2$, while the alphabet size increases by the same factor $2$.

%%%%%%%%%%%%%%%%%%%%%%%%%%%%%%%%%%%%%%%%%%%%%%%%%%%%%%%%%%

\subsection{Arbitrary divisions}
\label{sub:arbitrary}

For simplicity, and for explaining the divide-and-conquer technique, in Subsection~\ref{sub:divide} we divided the set of users in $2$ groups of roughly equal size. This can easily be generalized to splitting the users in $k \geq 2$ groups. For simplicity, let us assume that both $n$ and $q$ are divisible by $k$, and that $q = k q_0$ for some $q_0$. Let us denote the random variable describing the distribution of colluders among the groups by a vector $\vec{C} = (C^{(1)}, \ldots, C^{(k)})$, with $C^{(t)}$ being the number of colluders assigned to group $U^{(t)}$. Then, for each $t$, the random variable $C^{(t)}$ follows a hypergeometric distribution with mean $\frac{c}{k}$ and variance less than $\frac{c}{k}$. 

Similar to the fact that the tails of the hypergeometric distribution are smaller than the tails of the binomial distribution, it can be shown that the probability that $\max_{1\leq t \leq k} C^{(t)}$ exceeds some value $a$ is smaller than the probability that the maximum entry $\max_{1 \leq t \leq k} M^{(t)}$ of a uniform multinomial random variable $\vec{M} = (M^{(1)}, \ldots, M^{(k)})$ exceeds the same value $a$. This allows us to apply a result from Raab and Steger \cite[Theorem 1]{raab98}, which says that for values $k$ such that $k \ln k = o(c)$, this maximum $\max_{1 \leq t \leq k} M^{(t)}$ is always very close to its mean $\frac{c}{k}$. More precisely, for $\alpha_k = O\left(\sqrt{\ln\frac{k}{\eps_2}}\right)$ for some $\eps_2 > 0$, with high probability the group with the largest number of colluders will not contain more than $\frac{c}{k} + \alpha_k \sqrt{\frac{c}{k}}$ colluders:
\begin{align*}
P\left(\max_{1\leq t \leq k} C^{(t)} > \frac{c}{k} + \alpha_k \sqrt{\frac{c}{k}}\right) \leq \frac{\eps_2}{2}.
\end{align*}
So after splitting the users in $k$ groups of size $\frac{n}{k}$, we know that with probability at least $1 - \frac{\eps_2}{2}$ each group contains at most $\frac{c}{k} + \alpha_k\sqrt{\frac{c}{k}}$ colluders. Then, for each group we independently generate $q_0$-ary traitor tracing schemes $(X^{(t)}, \sigma^{(t)})$ using $\mathcal{S}_{q_0}$, with parameters
\begin{align*} 
\left\{\begin{array}{ll}
c^{(t)} = \dfrac{c}{k} + \alpha_k \sqrt{\dfrac{c}{k}}, \quad & \eps_1^{(t)} = \dfrac{\eps_1}{k}, \\ 
n^{(t)} = \dfrac{n}{k}, & \eps_2^{(t)} = \dfrac{\eps_2}{2k}.
\end{array}\right\} 
\quad (t = 1, \ldots, k)
\end{align*}
The probability that the splits go well and the tracing of traitors in each group goes well, is at least $(1 - \frac{\eps_2}{2}) (1 - \frac{\eps_2}{2k})^{k} \geq 1 - \eps_2$. The conquer-phase can then analogously be generalized to $k$ groups, weaving $k$ codes $X^{(t)}$ together to a big code $X$. We then end up with a $q$-ary traitor tracing scheme with the following properties.

\begin{theorem} \label{thm:thm2}
Let the $q$-ary traitor tracing scheme be as described above. Then, with probability at most $\eps_1$ at least one innocent user is caught, and with probability at most $\eps_2$ not all pirates are disconnected after at most $\ell = \sum_{t=1}^k \ell^{(t)}$ segments.
\end{theorem}

So if we can construct $q_0$-ary schemes with codelengths $\ell_{q_0}(c,n,\eps_1,\eps_2)$ quadratic in $c$, then the divide-and-conquer technique provides us with $q$-ary schemes ($q = k q_0$) with codelengths
\begin{align*}
\ell_q(c,n,\eps_1,\eps_2) = k\ell_{q_0}\left(\frac{c}{k} + \alpha_k \sqrt{\frac{c}{k}}, \frac{n}{k}, \frac{\eps_1}{k}, \frac{\eps_2}{2k}\right) \\ \approx k\ell_{q_0}\left(\frac{c}{k}, n, \eps_1, \eps_2\right) \approx \frac{1}{k} \ell_{q_0}\left(c, n, \eps_1, \eps_2\right).
\end{align*}
So the codelength decreases by a factor of approximately $k$, while the alphabet size increases by the same factor $k$. In particular, using a binary scheme with a codelength of $\ell_2$ quadratic in $c$ as a starting point, we obtain $q$-ary traitor tracing schemes with codelengths satisfying
\begin{align}
\ell_q(c,n,\eps_1,\eps_2) \approx \frac{2}{q} \ell_2(c,n,\eps_1,\eps_2).
\end{align}

\subsubsection*{Remark} For explaining the divide-and-conquer method, we assumed the smaller codes $X^{(t)}$ were generated in advance, i.e., during the divide-phase. This is not necessary, as one could also generate the new symbols for users on the fly, once they are needed. In practice, one may not want to generate all codewords in advance, but let them depend on the previous pirate ouput. Then, only when $y_{i-1} \in Q^{(t)}$ for some $t$ is known, the distributor generates new symbols for users $j \in U^{(t)}$. This means that this divide-and-conquer method works for any probabilistic dynamic traitor tracing scheme, even when the codewords cannot be generated in advance. 

%%%%%%%%%%%%%%%%%%%%%%%%%%%%%%%%%%%%%%%%%%%%%%%%%%%%%%%%%%
%%%%%%%%%%%%%%%%%%%%%%%%%%%%%%%%%%%%%%%%%%%%%%%%%%%%%%%%%%

\section{The $q$-ary dynamic Tardos scheme}
\label{sec:tardos}

Recently, Laarhoven et al.\ \cite{laarhoven11b} showed that one can efficiently turn the binary static Tardos scheme \cite{tardos03}, or any variant thereof \cite{skoric08,blayer08,laarhoven11}, into a dynamic scheme that is able to catch \textit{all} colluders with a codelength that is quadratic in $c$. More precisely, for $q = 2$ and parameters $c, n, \eps_1, \eps_2$, one can create schemes $(X, \sigma)$ with codelengths $L_2$ satisfying
\begin{align}
L_2(c,n,\eps_1,\eps_2) = \left[\frac{\pi^2}{2} + O\left(\sqrt[3]{\frac{1}{c}}\right)\right] c^2 \ln\left(\frac{n}{\eps_1}\right). \label{eq:tardos}
\end{align}
Note that the codelength does depend on $\eps_2$, but $\eps_2$ only appears in lower order terms; see Laarhoven et al.\ \cite{laarhoven11b} for details. Using this construction as our `base construction' $\mathcal{S}_2$, the divide-and-conquer construction allows us to construct $q$-ary dynamic Tardos schemes with the following codelengths $L_q$.

\begin{theorem} \label{thm:thm3}
For arbitrary (even) $q$ satisfying $q \ln q = o(c)$, we can construct $q$-ary dynamic Tardos schemes with codelengths $L_q$ given by
\begin{align}
& L_q(c,n,\eps_1,\eps_2) \nonumber \\ 
 &= \left[\pi^2 + O\left(\sqrt[3]{\frac{q}{c}} + \sqrt{\frac{q \log \frac{q}{\eps_2}}{c}}\right)\right] \frac{c^2}{q} \ln\left(\frac{n}{\eps_1}\right). \label{eq:Lq}
\end{align}
\end{theorem}

\begin{proof}
Let $q = 2k$ be even. Combining Equation~\eqref{eq:tardos} with Theorem~\ref{thm:thm2}, we get
\begin{align*}
& L_q(c,n,\eps_1,\eps_2) = k L_2\left(\frac{c}{k}+\alpha_k \sqrt{\frac{c}{k}},\frac{n}{k},\frac{\eps_1}{k},\frac{\eps_2}{2k}\right) \\ 
 &= k \left[\frac{\pi^2}{2} + O\left(\left(\frac{c}{k} + \alpha_k \sqrt{\frac{c}{k}}\right)^{-1/3}\right)\right] \\
 &\quad \cdot \left[\frac{c}{k} + \alpha_k\sqrt{\frac{c}{k}}\right]^2 \ln\left(\frac{n/k}{\eps_1/k}\right)
\end{align*}
Since $\ln k = o(\frac{c}{k})$, we have $\sqrt{\frac{c}{k} \ln \frac{k}{\eps_2}} = o\left(\frac{c}{k}\right)$, so the order term above simplifies to $O\left(\sqrt[3]{\frac{k}{c}}\right)$. Expanding the square, and observing that the product of the order-terms is small compared to the cross-terms, we get
\begin{align*}
& L_q(c,n,\eps_1,\eps_2) \\
 &= \left[\frac{\pi^2}{2} + O\left(\sqrt[3]{\frac{k}{c}}\right)\right] \left[\frac{c^2}{k} + O\left(c\sqrt{\frac{c \ln \frac{k}{\eps_2}}{k}}\right)\right] \ln\left(\frac{n}{\eps_1}\right) \\
 &= \left[\frac{\pi^2}{2} + O\left(\sqrt[3]{\frac{q}{c}}\right)\right] \left[2 + O\left(\sqrt{\frac{q \ln \frac{q}{\eps_2}}{c}}\right)\right] \frac{c^2}{q} \ln\left(\frac{n}{\eps_1}\right) \\
 &= \left[\pi^2 + O\left(\sqrt[3]{\frac{q}{c}} + \sqrt{\frac{q \ln \frac{q}{\eps_2}}{c}}\right)\right] \frac{c^2}{q} \ln\left(\frac{n}{\eps_1}\right).
\end{align*}
This is exactly Equation~\eqref{eq:Lq}.
\end{proof}

\subsubsection*{Remark} In Equation~\eqref{eq:Lq}, the first order term contains two terms. For small values of $q$ compared to $c$, the first of these terms $O(\sqrt[3]{\frac{q}{c}})$ dominates, as the third root is larger than the square root term. However, for $q \ln q$ close to $O(c)$ and large $q$ and $c$, the second term $O\left(\sqrt{\frac{q \ln \frac{q}{\eps_2}}{c}}\right)$ will start to dominate. So which of these terms is bigger depends on the relation between $q$ and $c$.

Asymptotically, the codelengths $L_q$ in Equation~\eqref{eq:Lq} are a factor $q/2$ shorter than the codelengths of the binary dynamic Tardos scheme. These codelengths also match the static fingerprinting capacity as obtained by Boesten and \v{S}kori\'{c} \cite{boesten11} and Huang and Moulin \cite{huang12}, up to a constant factor. Since we are considering a dynamic setting, this does not mean that these codelengths are optimal, but it does show that converting any $q$-ary static Tardos scheme to a $q$-ary dynamic Tardos scheme via Laarhoven et al.'s construction \cite{laarhoven11b} will at best lead to the same asymptotic codelengths. Figure~\ref{fig1} shows variants of the Tardos scheme, and ways to construct them. To construct a $q$-ary dynamic Tardos scheme from the optimal binary static Tardos scheme of Laarhoven and De Weger \cite{laarhoven11}, one has to (i) make the scheme dynamic, and (ii) go from a binary to a $q$-ary alphabet. First applying (ii) from \v{S}kori\'{c} et al.\ \cite{skoric08} and then applying (i) using the construction of Laarhoven et al.\ \cite{laarhoven11b} leads to codes that are a factor $O(\ln q) \cdot O(1) = O(\ln q)$ shorter. We showed that first applying (i) using the construction of Laarhoven et al.\ \cite{laarhoven11b} and then applying (ii), we get codes that are a factor $O(1) \cdot O(q) = O(q)$ shorter.

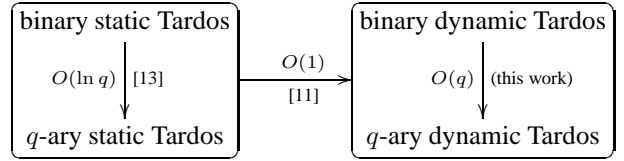
\begin{figure}
\centering
\begin{align*}
\def\g#1{\save[].[d]!C="g#1"*[F-:<3pt>]\frm{}\restore}%
\xymatrix@C=1.5cm@R=1cm{
\g1 \text{binary static Tardos} \ar[d]_{O(\ln q)}^{\text{\cite{skoric08}}} & \g2 \text{binary dynamic Tardos} \ar[d]_{O(q)}^{\text{(this work)}} \\
\text{$q$-ary static Tardos} & \text{$q$-ary dynamic Tardos}
\ar"g1";"g2"^{O(1)}_{\text{\cite{laarhoven11b}}}
}
\end{align*}
\caption{Known variants of the Tardos scheme, and a comparison of their asymptotic codelengths. For static Tardos schemes, current methods to construct $q$-ary schemes \cite{skoric08} lead to codelengths that are a factor $O(\ln q)$ shorter than binary schemes. For dynamic Tardos schemes, we showed that with the divide-and-conquer approach, the codelength decreases by a factor of $O(q)$. From any static Tardos scheme, one can obtain a dynamic Tardos scheme with the same alphabet size and the same order codelengths \cite{laarhoven11b}.\vspace{-0.4cm}\label{fig1}}
\end{figure}

%%%%%%%%%%%%%%%%%%%%%%%%%%%%%%%%%%%%%%%%%%%%%%%%%%%%%%%%%%

\subsection{Large-$q$ asymptotics}
\label{sub:asymptotics}

Instead of considering the asymptotic behaviour of fixed $q$ and large $c$, one could also consider the asymptotic behaviour of large $q$ and $c$. For instance, if we let $q = O(c^{1-\gamma})$ we get the following corollary.

\begin{corollary} \label{cor:cor1}
Let $q = O(c^{1-\gamma})$ for some $\gamma > 0$. Then by Theorem~\ref{thm:thm3}, we can construct $q$-ary dynamic Tardos schemes achieving asymptotic codelengths of
\begin{align}
L_q(c,n,\eps_1,\eps_2) = \left[\pi^2 + O(c^{-\gamma/3})\right] c^{1+\gamma} \ln\left(\frac{n}{\eps_1}\right).
\end{align}
\end{corollary}

For $\gamma \to 0$, we get alphabet sizes $q$ almost linear in $c$, so it makes sense to compare this construction to the deterministic schemes of Fiat and Tassa \cite{fiat01} and Berkman et al.\ \cite{berkman01}. The optimal scheme of Berkman et al.\ uses an alphabet size of $q = c + 1$, and requires a codelength of $\ell_q(c,n) = O(c^2 + c \log_2 n)$. For large $c$, this scheme therefore requires longer codes \textit{and} larger alphabets than the $q$-ary dynamic Tardos scheme. The scheme of Fiat and Tassa uses an alphabet of size $q = 2c + 1$, and requires a codelength of only $\ell_q(c,n) = c \log_2 n + c$. Our scheme approaches this asymptotic codelength as $\gamma \to 0$, but the constants of our scheme are larger, and of course Fiat and Tassa's scheme is deterministic. So, if one can afford using an alphabet of size $q = 2c + 1$, Fiat and Tassa's scheme is clearly the way to go, but for lower values of $q$, the $q$-ary dynamic Tardos scheme seems to be the best asymptotic scheme known so far.

%%%%%%%%%%%%%%%%%%%%%%%%%%%%%%%%%%%%%%%%%%%%%%%%%%%%%%%%%%

\subsection{The universal Tardos scheme}
\label{sub:univtardos}

Besides the dynamic Tardos scheme, Laarhoven et al.\ \cite[Section V]{laarhoven11b} also show how to efficiently catch coalitions of a priori unknown sizes $c$, using a variant of the dynamic Tardos scheme known as the \textit{universal} Tardos scheme. With slightly longer codelengths and maintaining multiple accusation scores per user, one can guarantee that small coalitions are caught much faster. The divide-and-conquer construction can trivially be applied to this variant as well. In this case, the practical difficulty of bounding the number of colluders in each group even disappears, since the universal Tardos scheme does not require the distributor to provide values of $c$ anymore. One simply divides the set of users in $k$ groups, and assigns the parameters $(n, \eps_1, \eps_2)^{(t)} = (\frac{n}{k}, \frac{\eps_1}{k}, \frac{\eps_2}{k})$ to each group. Then, one can easily show that the scheme will catch any coalition with a codelength quadratic in the actual number of colluders. 

%%%%%%%%%%%%%%%%%%%%%%%%%%%%%%%%%%%%%%%%%%%%%%%%%%%%%%%%%%
%%%%%%%%%%%%%%%%%%%%%%%%%%%%%%%%%%%%%%%%%%%%%%%%%%%%%%%%%%

\section{Summary}
\label{sec:summary}

We have shown that with the divide-and-conquer approach, we can obtain schemes for alphabet sizes $q = kq_0$ which have codelengths approximately equal to the sum of $k$ times the codelength of a $q_0$-ary traitor tracing scheme. Applying this to the binary dynamic Tardos scheme of Laarhoven et al.\ \cite{laarhoven11b}, this leads to codelengths which are quadratic in the number of colluders $c$ and decreasing linearly in the alphabet size $q$. Thus, the codelengths of this construction match the $q$-ary static fingerprinting capacity of Boesten and \v{S}kori\'{c} \cite{boesten11} and Huang and Moulin \cite{huang12}, up to a constant factor. For $q$ growing almost linearly in $c$, the codelengths approach the asymptotic codelengths of Fiat and Tassa \cite{fiat01}, and improve upon the codelengths of Berkman et al.\ \cite{berkman01}.

There are several interesting open problems for future research in this area. We mention some below.

\subsection{The capacity of the dynamic fingerprinting game} To the best of our knowledge, no one has yet investigated whether the fingerprinting capacity game can be extended to the dynamic traitor tracing setting. Above, we compared our codelengths obtained from the dynamic Tardos scheme to the static fingerprinting capacity, but it would be more interesting to be able to compare these codelengths to (bounds on) the dynamic capacity. The above construction does make a start in this direction, by showing that the $q$-ary dynamic fingerprinting capacity is at least a factor $q/2$ higher than the binary dynamic fingerprinting capacity.

\subsection{The $q$-ary static Tardos scheme} \v{S}kori\'{c} et al.\ \cite{skoric08} gave a construction for $q$-ary Tardos codes, which are roughly a factor $O(\ln q)$ shorter than binary Tardos codes. It would be interesting to see if it is possible to construct $q$-ary Tardos codes which are a factor $O(q)$ shorter and approach the $q$-ary fingerprinting capacity. With the dynamic Tardos construction of Laarhoven et al.\ \cite{laarhoven11b} and our current results, this may then also lead to better dynamic traitor tracing schemes.

\subsection{Application to different schemes} Above we showed that our construction can be applied to the binary dynamic Tardos scheme, but we can also apply our results to $q_0$-ary dynamic Tardos schemes, or a completely different binary dynamic traitor tracing scheme. If someone finds better binary dynamic schemes, combined with our construction this would immediately lead to better $q$-ary dynamic traitor tracing schemes.

\subsection{Variants of the divide-and-conquer construction} One can think of many variants of the divide-and-conquer scheme, but these seem harder to analyze. For example, instead of using disjoint sets of symbols for each group, one could let the different alphabets overlap in a few symbols. Or, instead of always using the same division of colluders in groups, one may want to redo the division of users in groups for every position, or every time a user is disconnected. Analyzing these variants may lead to further improvements.

\end{document}